\documentclass[a4paper]{article}

\usepackage[margin=20truemm]{geometry}
\usepackage[indLines=false]{algpseudocodex}
\usepackage{amsmath,amssymb,amsfonts,amsthm,algorithm,mathbbol,tabularx,refcount}
\usepackage[capitalise, noabbrev]{cleveref}
\theoremstyle{plain}
\newtheorem{theorem}{Theorem} 
\newtheorem{lemma}{Lemma} 
\newtheorem{definition}{Definition} 
\newtheorem{claim}{Claim}
\crefname{claim}{Claim}{Claims}

\title{Computational Power of a Single Oblivious Mobile Agent in Two-Edge-Connected Graphs} 
\author{Taichi Inoue$^\dagger$\and Naoki Kitamura$^\dagger$ \and Taisuke Izumi$^\dagger$ \and Toshimitsu Masuzawa$^\dagger$ \\ \\ {\small $\dagger$ Osaka University, 1-5, Yamadaoka, Suita, Osaka, 565-0871, Japan}}
\date{}

\begin{document}

\maketitle

\begin{abstract}
  We investigated the computational power of a single mobile agent in an $n$-node graph with storage (i.e., node memory).
  Generally, a system with one-bit agent memory and $O(1)$-bit storage is as powerful as that with $O(n)$-bit agent memory and $O(1)$-bit storage.
  Thus, we focus on the difference between one-bit memory and oblivious (i.e., zero-bit memory) agents.
  Although their computational powers are not equivalent, all the known results exhibiting such a difference rely on the fact that oblivious agents cannot transfer any information from one side to the other across the bridge edge.
  Hence, our main question is as follows: Are the computational powers of one-bit memory and oblivious agents equivalent in 2-edge-connected graphs or not?
  The main contribution of this study is to answer this question under the relaxed assumption that each node has $O(\log\Delta)$-bit storage (where $\Delta$ is the maximum degree of the graph).
  We present an algorithm for simulating any algorithm for a single one-bit memory agent using an oblivious agent with $O(n^2)$-time overhead per round.
  Our results imply that the topological structure of graphs differentiating the computational powers of oblivious and non-oblivious agents is completely characterized by the existence of bridge edges.
\end{abstract}

\section{Introduction}
\subsection{Background and Our Result}
A \emph{mobile agent} (hereinafter called an \emph{agent}) is an individual entity that performs a given task by autonomously moving in a graph (or network).
This is one of the main computational paradigms of distributed algorithms.
In the theory of mobile agent systems, there exist various models that differ in memory resources, asynchrony, observation capability of the system, and so on.
Revealing the computational power of each model is recognized as a central question in this research field.
This study focuses on the computational power of a system with a single agent, despite that it is fairly commonplace to consider the cooperation of multiple agents.

In single-agent systems, the amount of (persistent) memory held by the agent and nodes is a major quantitative parameter of their computational capability.
The computational cost is measured by \emph{rounds} (i.e., the number of movements by the agent).
Throughout this paper, we refer to the memory of the agent as \emph{memory} and that of each node as \emph{storage}.
It is almost evident that an agent with sufficiently large memory can simulate any centralized algorithm.
That is, any computable problem can be solved.
Hence, theoretical interest lies in problem solvability when the amount of memory is limited.
For example, the graph exploration problem, which requires the agent to visit all nodes in the graph, is a fundamental problem
Generally, $\Theta(\log n)$-bit memory is necessary and sufficient to solve this problem when the nodes have no storage \cite{DBLP:journals/tcs/FraigniaudIPPP05,DBLP:journals/jacm/Reingold08}.
It is also proved that the system with $O(1)$-bit memory and $O(1)$-bit storage can solve the graph exploration problem \cite{DBLP:journals/talg/CohenFIKP08}.

In this study, we investigated the effect of the amount of memory on problem solvability.
A known result on this research line is that, in any $n$-node graph with $O(1)$-bit storage per node, the agent with one-bit memory is as powerful as that with $O(n)$-bit memory \cite{https://doi.org/10.48550/arxiv.2209.01906}.
More precisely, any property of the graph $G$ decidable by the $O(n)$-bit memory agent within polynomial movements and polynomial-time local computation is also decidable by the one-bit memory agent within polynomial movements and polynomial-time local computation.
This result is not limited to decision problems, but actually provides a general technique for simulating the execution of any $O(n)$-bit memory agent using a one-bit memory agent with a polynomial-time multiplicative overhead per round.
It has also been shown that the computational powers of the one-bit and zero-bit (i.e., oblivious) memory agents are inherently different regardless of the amount of available storage.
Thus, our question is already closed in general settings.  
However, whether the separation between the one-bit memory and oblivious agents is exhibited in a restricted graph class remains unclear. 
All known results exhibiting such a separation \cite{DBLP:journals/talg/CohenFIKP08,https://doi.org/10.48550/arxiv.2209.01906} rely on the existence of a bridge edge (i.e., a single edge such that its removal disconnects the graph) in the graph.
More precisely, they are derived from the fact that the oblivious agent cannot transfer any information from one side to the other side across the bridge edge, whereas the one-bit memory agent can. 
Therefore, our central question is then rephrased as:
Are the computational powers of the one-bit memory and oblivious agents equivalent in 2-edge-connected graphs or not?
The main contribution of this study is to answer this question positively.
We focus on 2-edge-connected graphs with a maximum degree $\Delta$ and relax the constraint on the amount of storage from $O(1)$ bits to $O(\log\Delta)$ bits.
In this setting, we present a polynomial-time overhead algorithm that simulates the execution of a single one-bit memory agent by an oblivious agent.
By combining the results of \cite{https://doi.org/10.48550/arxiv.2209.01906}, we can deduce the equivalence between an $O(n)$-bit memory agent and an oblivious agent in 2-edge-connected graphs.
This implies that the topological structure of graphs differentiating the computational powers of oblivious and non-oblivious agents is completely characterized by the existence of bridge edges.
The authors believe that this result provides a sharp insight into the computational complexity theory of mobile agents. 

\subsection{Technical Idea}
The model considered in this study follows the standard assumptions on mobile agent systems.
The graph $G$ is \emph{anonymous} (i.e., the agent cannot refer to the unique IDs of nodes), and the neighbors of a node are identified by the \emph{local port numbers} assigned to the edges incident to the node.
When an agent enters a node, it recognizes the \emph{entry port number} (assigned to the edge used to visit the current node).
In one round, the agent performs local computation, which includes updating the information stored in its own memory and the storage at the current node, and decides the \emph{outgoing port number} (assigned to the edge used to leave the node).

The main technical issue in simulating a one-bit memory agent is how we can transfer the information to an oblivious agent.
We resolve this by utilizing information on entry port numbers.
To clarify our idea, we first consider the simple case in which the graph is an oriented cycle and assume that the agent at a node $s$ wants to transfer one-bit information to the neighbor $t$ of $s$.
There are two distinct paths from $s$ to $t$ (i.e., clockwise and counterclockwise, specified by the orientation), which result in different entry port numbers for $t$.
The agent can transfer an information bit $b \in \{0,1\}$ from $s$ to $t$ using the path ending up with port number $b$.

When the graph topology is arbitrary, the simulation algorithm becomes more complicated; however, the fundamental idea is the same.
The algorithm first finds a cycle $C$ of $G$ containing edge $(s,t)$, which we refer to as the \emph{information transfer cycle (ITC)}, to transfer one-bit information from node $s$ to its neighbor $t$.
The algorithm provides a consistent orientation to the ITC and applies the algorithm above for oriented cycles to the ITC.
Because we assume that $G$ is 2-edge-connected, there necessarily exists an ITC for any edge $(s,t)$ in $G$.

Our algorithm runs a depth-first search (DFS) to compute ITC.
The agent starts DFS from $s$ with the choice of $(s, t)$ as the first traversal edge.
Because $G$ is 2-edge-connected, the agent eventually reaches $s$ again through edge $e$ different from $(s, t)$.
\cref{fig:ITC} shows an example of transmission on an ITC.
\begin{figure}[!t]
  \centering
    \includegraphics[keepaspectratio,scale=0.825]{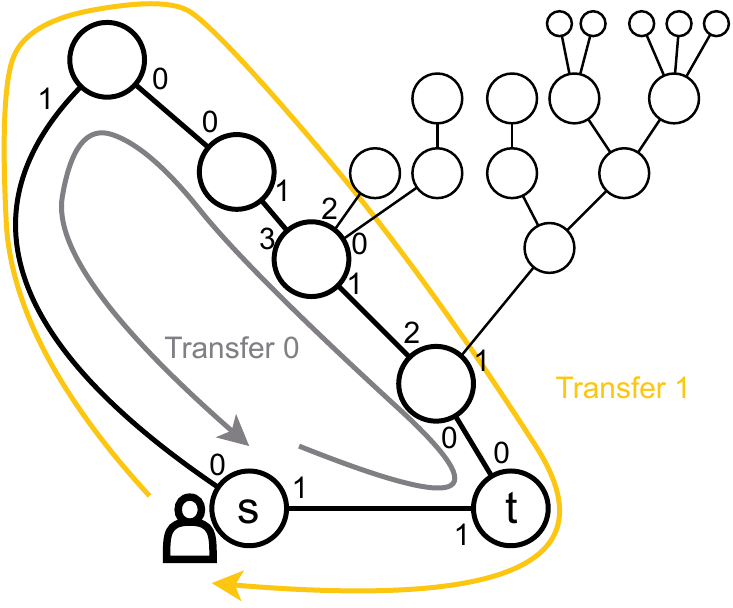}
    \caption{Example of the construction of an ITC (drawn by bold lines) on a DFS tree and imitative transmission of one-bit data. Arrows are moving directions of the agent. Some local port numbers are omitted.}
    \label{fig:ITC}
\end{figure}
Let $P_{s, t}$ be the path from $s$ to $t$ managed by the DFS algorithm.
Then, the algorithm obtains the ITC $P_{s,t} + e$.  
The agent must reset all information because the garbage information is left at the nodes that are traversed by the DFS process but not contained in the constructed ITC.
It is implemented by careful re-execution of the DFS starting from $s$ with the first traversal edge $(s, t)$.

We emphasize that its implementation is far from triviality, although the intuition of our algorithm stated above is simple and easy to follow.
A technical hurdle is that the agent itself is oblivious.
The agent cannot memorize the subtask currently executed despite our algorithmic idea being constituted of several subtasks.
Our algorithm conducts the sequential phase-by-phase composition of all subtasks carefully, using only node storage.
Details of the design are presented in \cref{ouralg}.

\subsection{Related Work}
To the best of our knowledge, the authors' prior work \cite{https://doi.org/10.48550/arxiv.2209.01906} is the first to consider the effect of memory size on the solvability of general tasks in graphs with storage, but several studies have been conducted on specific problems.
One benchmark problem is the graph exploration problem.
In the case where nodes have no storage, it has been shown that $\Theta(\log n)$-bit memory is necessary and sufficient to solve it within polynomial time \cite{DBLP:journals/tcs/FraigniaudIPPP05,DBLP:journals/jacm/Reingold08}.
When allowing nodes to be stored, the graph exploration algorithm can be solved by an oblivious agent using $O(\log \Delta)$-bit storage per node (where $\Delta$ is the maximum degree of the graph) \cite{DBLP:journals/ieicet/SudoBNOKM15}.
Minimizing the storage size, there exists a graph exploration algorithm for an $O(1)$-bit memory agent using $O(1)$-bit storage \cite{DBLP:journals/talg/CohenFIKP08}.
It also shows that exploration using a single oblivious agent is impossible if the storage size is $O(1)$ bits.
Note that the topology used in their impossibility proof has a bridge edge; that is, it is not 2-edge-connected.
An interesting open problem inspired by our result is whether an oblivious agent with $O(1)$-bit storage can solve the graph exploration problem for 2-edge-connected graphs or not.
It is possible to explore 2-edge-connected graphs if one can construct the ITC using only $O(1)$-bit storage because our simulation algorithm requires $\omega(1)$-bit storage only during the construction of the ITC.
There also exists a model in which the agent uses a ``pebble'' device that can be put on and picked up from nodes by the agent.
Disser et al. \cite{DBLP:conf/soda/DisserHK16} showed that $n$-node graph exploration can be solved using $\Theta(\log\log n)$ pebbles.
However, this method has the disadvantage of being time inefficient since it takes $n^{O(\log \log n)}$ rounds to complete the exploration.
Although it is a crucial assumption in our simulation algorithm that the agent knows the entry port number, there is a weaker model called the \emph{myopic robot}, in which the entry port number is undetectable by agents  \cite{DBLP:conf/ipps/DattaLLP15,DBLP:conf/icdcn/Kamei21,DBLP:conf/opodis/KameiLOTW19,DBLP:journals/iandc/OoshitaT22}.
Although no formal proof is given, it is almost trivial that such a weaker agent cannot simulate the agent memory as in our result, even under the assumption of 2-edge-connectivity.
While the model and problem are significantly different, Dieudonne et al. \cite{DBLP:journals/ijfcs/DieudonneDP019} presented a solution based on an approach similar to ours.
In that study, the authors considered the information transfer between the robots in the 2D plane, which cannot have an explicit communication channel. 
In their approach, the transferred information is not embedded into movement patterns as in our approach but is embedded into the (real-value) distance between two robots.

\section{Preliminaries}
\subsection{Graph}
We used $[a,b]$ to denote the set of integers at least $a$ and at most $b$.
Let $G$ be a graph with $n$ nodes. 
This is formally defined as a port-numbered graph $G=(V,E,\Pi)$, where $V$ is the set of nodes, $E$ is the set of edges, and $\Pi$ is the set of port-numbering functions (explained later).
Throughout this paper, we assume that $G$ is simple, undirected, and 2-edge-connected.
Each node in $G$ is identified by a non-negative integer $i\ (0\le i\le n-1)$.
However, the agent operating in $G$ cannot see those values, that is, $G$ is anonymous (the behavior of the agent is formally defined in \cref{agent}).
For $e\in E$, we define $G-e$ as a graph obtained by removing $e$ from $G$.
Note that $G-e$ is always connected because we assume that $G$ is 2-edge-connected. 
The edges incident on each node $i$ are distinguished by \emph{local port numbers}.
The \emph{port numbering function} at node $i$ is defined as $\pi_i:[0,\Delta_i-1]\to N_i$, where $\Delta_i$ is the degree of node $i$ and $N_i$ is the set of neighbors of $i$.
The maximum degree of $G$ is defined as $\Delta$.
Note that $\Delta_i$ is necessarily greater than one because $G$ is 2-edge-connected.
We define $\pi^{-1}_i$ as the inverse function of $\pi_i$.
The set $\Pi$ consists of port-numbering functions for all the nodes in $V$.
Each node in $G$ has storage (i.e., node memory) that keeps the stored information even after the agent leaves the node.

\subsection{Mobile Agent}
\label{agent}
We consider a single-agent model in which one mobile agent moves in graph $G$.
A mobile agent has two distinct memory spaces: \emph{temporary memory} and \emph{persistent memory}.
Temporary memory is a working space used for computation at each node, but the stored information is completely reset when the agent leaves the node.
In contrast, persistent memory space retains the stored information even when the agent moves from one node to another.
In the following argument, we use the terminology ``(agent) memory'' to refer to persistent memory and measure the space complexity of the agent by the amount of persistent memory.
Although the size of the available temporary memory is unbounded, $O(\log \Delta)$ bits of temporary memory are sufficient to implement our algorithm.
In this study, we focus on the \emph{oblivious agent} (or the agent with zero-bit memory) and \emph{one-bit agent} (or the agent with one-bit memory).

The execution of the agent follows discrete rounds $t=0,1,\cdots$.
In each round, the agent performed the following two computational steps: 
\begin{enumerate}
  \item Let $i$ be the node where the agent currently stays.
  First, the agent starts the local computation specified by the algorithm.
  At the beginning of the local computation, the local port number corresponding to the edge through which it has entered the current node $i$ (i.e., \emph{entry port number}) and the degree $\Delta_i$ of $i$ are stored in the temporary memory.
  Note that the entry port number can be arbitrary at the beginning of the execution.
  Following the information stored in the  temporary memory, persistent memory, and storage of $i$, the agent updates the storage of $i$ and its own persistent memory.
  It also determines the port number corresponding to the edge through which it leaves the node $i$ (i.e., \emph{outgoing port number}).
  \item The agent moves to the neighbor of $i$ specified by the outgoing port number.
  Note that the agent is guaranteed to arrive at the neighbor within the current round.
\end{enumerate}
\cref{fig:round} illustrated the behavior of a single round.
\begin{figure}[t]
  \centering
  \includegraphics[keepaspectratio,width=\linewidth]{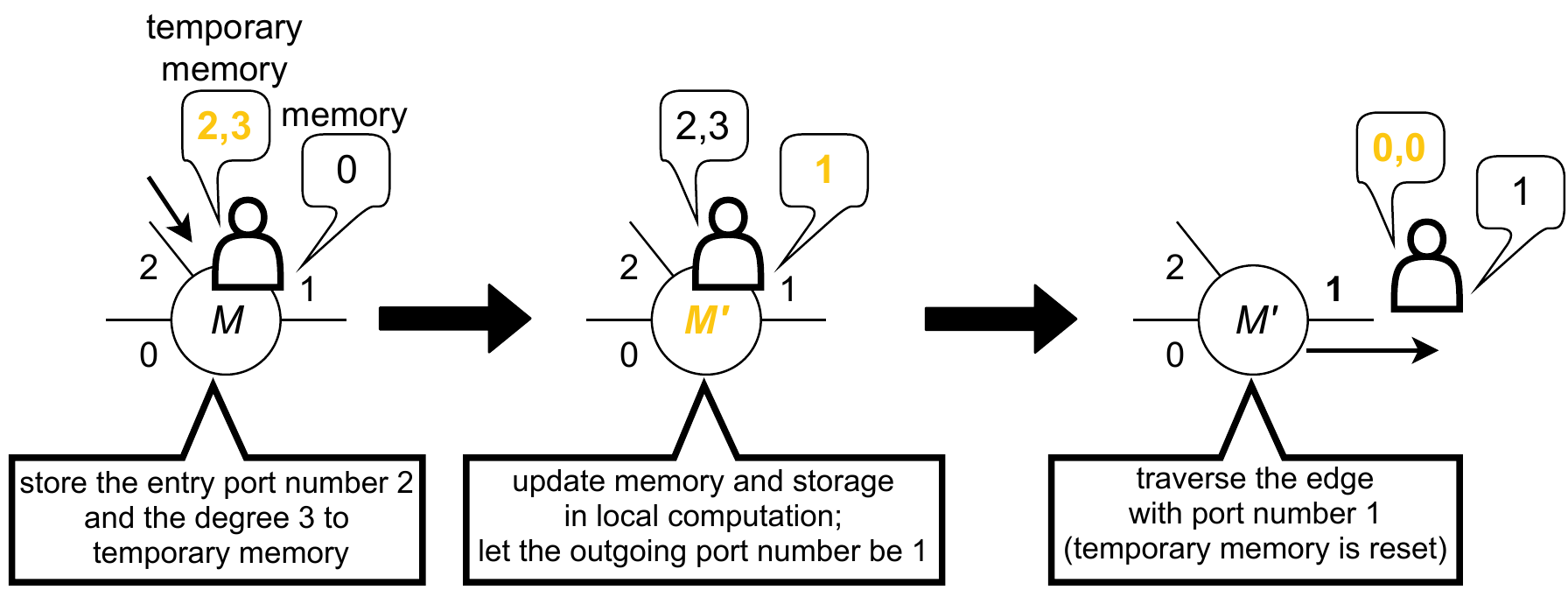}
  \caption{Workflow of one-round operation (where $M$ and $M'$ mean the storage value of the node)}
  \label{fig:round}
\end{figure}
An algorithm $A$ for a one-bit agent is defined as $A=(m,M,M',l,\phi)$, a 5-tuple of an initial memory value $m$, an initial storage value $M$ for each node, an initial storage value $M'$ for the initial location of the agent, an initial location $l$ of the agent, and a transition function $\phi$.
Note that the special initialization value $M'$ is crucial for implementing the simulation because a symmetry-breaking mechanism is required to distinguish the initial location of the simulated agent from other nodes. 
The \emph{transition function $\phi$} for a one-bit agent is defined as $\phi:\mathbb{N}\times \mathbb{Z}_{\ge-1}\times \{0,1\}^{\ast}\times \left\{0,1\right\} \to \mathbb{Z}_{\ge-1}\times \{0,1\}^{\ast}\times \left\{0,1\right\}$, where  $\mathbb{Z}_{\ge -1}$ is a set of integers greater than or equal to -1.
The arguments given to $\phi$ are the degree of the node, entry port number, storage value of the current node, and memory value.
Note that the bit length of the storage is not fixed by the algorithm because our model allows the storage size to depend on the parameters of graph $G$.
The returned values are the outgoing port number, storage value to be left at the current node, and memory value after local computation.
Outgoing port number -1 implies that the agent terminates the execution of the algorithm and stays at the same node.
For convenience, we consider an algorithm for an oblivious agent such that the memory value given to and returned by $\phi$ is always 0.
Note that $\phi$ does not consider the identifier of the current node as an argument because we assume anonymous graphs.

Next, we define the system configuration and execution.
\begin{definition}[configuration]
  A \emph{configuration $C$} of the system is represented as $C=\left((M_{0},\allowbreak M_{1},\allowbreak \dots,\allowbreak M_{n-1}),\allowbreak m,\allowbreak p, l\right) \in \left\{\{0,1\}^{\ast}\right\}^{n}\times \{0,1\}\times \mathbb{Z}_{\ge -1}\times [0,n-1]$, where $M_i$ is the storage value of node $i$, $m$ is the memory value of the agent, $p$ is the entry port number, and $l$ is the location of the agent.
  A \emph{valid configuration} in $G$ is such that the value of the entry port number is in the range $[-1, \Delta_{l} - 1]$.
\end{definition}
\begin{definition}[execution]
  \label{def:exec}
  Let $A=(m,M,M',l,\phi)$ be an algorithm.
  Given two valid configurations, $C_t = ((M_{0,t}, \dots ,M_{n-1, t}), m_t, p_t, l_t)$ and $C_{t+1} = ((M_{0,t+1},\allowbreak \dots,\allowbreak M_{n-1, t+1}),\allowbreak m_{t+1},\allowbreak p_{t+1},\allowbreak l_{t+1})$, we consider that they \emph{follow} the transition function $\phi$ of $A$ in $G$, denoted by $C_{t}\xrightarrow{\phi,G}C_{t+1}$, if the following three conditions are satisfied:
  \begin{itemize}
    \item $\phi(\Delta_{l_t}, p_t, M_{l_t, t}, m_t) = \left(\pi^{-1}_{l_t}(l_{t+1}), M_{l_t, t+1}, m_{t+1}\right)$
    \item $M_{i,t+1} = M_{i,t}$ for all $i\neq l_t$
    \item $p_{t+1} = \pi^{-1}_{l_{t+1}}(l_t)$
  \end{itemize}
  Let $C$ be any valid configuration on $G$ such that the initial value of the memory is $m$, the location of the agent is $l$, the initial storage value for node $i \neq l$ is $M$ and $M'$ for node $l$.  
  The \emph{execution $E(G,C,A)$} of the system is defined as an infinite sequence of configurations $(C_t)_{t\ge0}$ such that $C_0=C$ and $C_{t}\xrightarrow{\phi,G}C_{t+1}$ holds for any $t\ge 0$.
\end{definition}
In \cref{def:exec}, any execution is defined as an infinite sequence; however, if the algorithm terminates in a finite round, then the execution is defined as an infinite sequence such that $C_t=C_{t+1}=\dots$ holds for the termination round $t$.
As the transition function $\phi$ is assumed to be deterministic, the execution of the algorithm is uniquely defined from the initial configuration $C_0=C$.
That is, $E(G,C,A)$ is uniquely defined.
In addition, throughout this study, we only consider algorithms that do not depend on the initial entry port number.
Then, execution $E(G,C,A)$ is determined only by algorithm $A$, graph $G$, and initial location $l$ of the agent.
Hence, we refer to the execution $E(G,C,A)$ as $E(G,v,A)$ with the initial location $v\in V$ of the agent. 

\subsection{Simulation}
Let $\mathcal{A}(k,\lambda)$ be the set of all algorithms that utilize $k$-bit memory and $\lambda$-bit storage.
Note that $k$ and $\lambda$ are not necessarily constant values, but might be functions that depend on the parameters of the graph.
The goal of this study was to design an algorithm $A\in \mathcal{A}(0,\lambda)$ that simulates any algorithm $A^\ast \in \mathcal{A}(1, \lambda^\ast)$, where $\lambda$ and $\lambda^{\ast}$ represent the storage sizes of the simulator and simulated algorithms, respectively.
To this end, we provide a formal definition for the simulation of the algorithm.
First, we present the notion of $\gamma$-configurations.
\begin{definition}[$\gamma$-configuration]
  \label{def:gamma}
  Let $\gamma = (\gamma_M, \gamma_A)$ be a pair of mappings, $\gamma_{M}:\{0,1\}^{\lambda}\to\{0,1\}^{\lambda^\ast}$ and $\gamma_{A}:\{0,1\}^{\lambda} \to\{0,1\}\times \mathbb{N}$.
  The \emph{$\gamma$-configuration $\gamma(C)$} is defined as $\gamma(C)=\left((\gamma_{M}(M_{0}),\gamma_{M}(M_{1}),\dots,\gamma_{M}(M_{n-1})\right),\allowbreak \gamma_{A}(M_{l}),l)$ for configuration $C=((M_{0},M_{1},\dots,\allowbreak M_{n-1}),\allowbreak m,p,l)$.
\end{definition}
Intuitively, the functions $\gamma_M$ and $\gamma_A$ work as the ``interpreter,'' which transforms a valid configuration of $A$ into that of $A^{\ast}$ under the assumption that the simulated agent executing $A^{\ast}$ stays at the same location as the simulator agent executing $A$.
The $\gamma$-configuration $\gamma(C)$ is the configuration of $A^{\ast}$ obtained by interpretation. 
It may seem unusual that $\gamma(C)$ does not explicitly depend on the value $m$ of the agent memory.
However, it is indispensable that algorithm $A$ stores the value of the agent memory of $A^{\ast}$ into the storage of the current node of $A$ to simulate the behavior of the simulated agent running $A^{\ast}$ by an oblivious agent.
The simulation algorithm is defined as follows:
\begin{definition}[simulation of algorithm]
  \label{def:simulation}
  Let $A^\ast\in \mathcal{A}(1,\lambda^\ast)$ be an algorithm.
  We consider that \emph{an algorithm $A\in \mathcal{A}(0,\lambda)$ simulates $A^{\ast}$ in $G$}, if there exists a pair of mappings $\gamma = (\gamma_M, \gamma_A)$ of \cref{def:gamma} that satisfy the following condition:
  \begin{quote}
    Let $v$ be any node of $G$, $E(G,v,A)=(C_t)_{t \geq 0}$, and $E(G,v,A^{\ast})=(C^{\ast}_t)_{t \geq 0}$.
    There uniquely exists the monotonically increasing sequence of round numbers $T(0),T(1),T(2),\dots$ such that $T(0)=0$ holds and $\gamma\left(C_{T(t)}\right)=C_t^\ast$ holds for each $t \ge 0$.
    That is, $E\left(G,v,A^\ast\right)=\gamma\left(C_{T(0)}\right),\gamma\left(C_{T(1)}\right), \gamma\left(C_{T(2)}\right),\dots$ holds.
  \end{quote}
\end{definition}

Finally, we defined the round complexity of the simulation as follows:
\begin{definition}
  Let $A \in \mathcal{A}(0,\lambda)$ be the algorithm that simulates $A^{\ast}\in \mathcal{A}(1,\lambda^\ast)$.
  Then, the time required to simulate round $t$ is defined as $T(t+1)-T(t)$.
  If $\max_{t\geq 0}\left(T(t+1)-T(t)\right)=\mathrm{poly}(n)$ holds for any $t \ge 0$, we consider that \emph{$A$ is a polynomial-time simulator of $A^{\ast}$ in $G$}. 
\end{definition}

\section{Our Algorithm}
\label{ouralg}
We show a polynomial-time simulator $A\in\mathcal{A}(0,\lambda)$ that simulates any one-bit agent algorithm $A^{\ast}\in\mathcal{A}(1,\lambda^\ast)$ in any 2-edge-connected graph $G$.
In the following argument, we denote the oblivious agent executing $A$ by the symbol $a$ and the one-bit simulated agent executing $A^{\ast}$ by $a^\ast$.
We also denote $A^{\ast} = \left(m^{\ast},M^\ast,(M')^\ast, l^\ast, \phi^{\ast}\right)$.
The storage size of algorithm $A$ is denoted by $\lambda$, which satisfies $\lambda=\lambda^{\ast}+O(\log\Delta)$.

\subsection{Mappings and Variables on Storage}
We first introduce all the variables stored in the storage of each node in \cref{tb:variables}.
\begin{table}[!t]
  \centering
  \caption{Variables used in our simulation algorithm. The initial location of $a$ (and $a^{\ast}$) is referred to as $v \in V$.}
  \begin{tabularx}{\linewidth}{lX}
    \hline
    Name & Role \\
    \hline
    \textit{sloc} & The binary flag representing the current location of $a^{\ast}$, i.e., $sloc_i = 1$ implies that the simulated agent $a^{\ast}$ is currently at node $i$. 
    At any time, at most one node $i$ satisfies $sloc_i = 1$.
    When $sloc_i = 0$ holds for every node $i$, the algorithm is simulating the movement of $a^{\ast}$.
    Initially, $sloc_v=1$ for the initial location $v \in V$ of $a$ and $a^{\ast}$, and $sloc_j=0$ for any $j \neq v$.\\
    \textit{smem} & The variable storing the current memory value of $a^{\ast}$ at the current node of $a^{\ast}$.
    It can store an arbitrary value at any other node.\\
    \textit{smemupd} & The flag indicating the completion of the information transfer along the constructed ITC.
    At the beginning of the one-round simulation of $A^{\ast}$, the value is 0 for all nodes.
    When $a$ propagates a binary memory value $m'$ of $a^\ast$ along the ITC, the node $i$ that has already received $m'$ sets $smemupd_i = 1$.
    The initial value is zero for all nodes.\\
    \textit{spin} & The entry port number of $a^\ast$ at the current node.
    It can store an arbitrary value at any other node.
    Initially, an arbitrary value is stored.\\
    \textit{spout} & The variable for storing the computed outgoing port number of $a^{\ast}$ at the current node.
    It can store an arbitrary value at any other node.\\
    \textit{svars} & The storage for the simulated algorithm $A^{\ast}$.
    The initial value is the one specified by the algorithm $A^{\ast}$.\\
    \textit{dfsstat} & The three-state variable used inside the DFS subroutine for constructing the ITC.
    If the agent $a$ does not execute the DFS subroutine, the value zero is stored at all nodes.
    The condition $dfsstat_i = 1$ implies that the agent $a$ has visited $i$, but $i$ still has the ports not probed yet, and $dfsstat_i = 2$ implies that all the ports of $i$ have been probed.
    The initial value is zero at all nodes.\\
    \textit{par} & The variable storing the port number to the parent of the DFS tree.
    After the construction of an ITC, this variable represents the predecessor of the ITC. 
    The initial value is arbitrary.\\
    \textit{cld} & The counter indicating the ports that have been already probed: $cld_i = x$ implies that up to the $(x-1)$-th port of the node $i$ has been probed.
    After the construction of an ITC, this variable represents the successor in the ITC.
    The initial value is zero for all nodes.\\
    \textit{sim} & The flag indicating the subtask in which the agent $a$ is engaged.
    If $sim_i = 0$ holds for all nodes $i$, the agent is executing the DFS subroutine for constructing an ITC.
    Once an ITC is constructed, $sim_i = 1$ holds if and only if $i$ is contained in the constructed ITC.\\
    \textit{lastin/out} & In execution of the DFS subroutine, these variables manage the latest entry/outgoing port numbers at each node.
    The stored information is used for deleting the ``garbage'' information in the storage left by the DFS subtask.\\
    \hline
  \end{tabularx}
  \label{tb:variables}
\end{table}
Each variable is referred to with an additional subscript to clarify the node that stores it.
For example, $sloc_i$ denotes the variable \textit{sloc} stored in node $i$.
For a tuple of variables $(X, Y)$ (where $X$ and $Y$ are the names of some variables), we also use the notation $(X, Y)_i$ to represent $(X_i, Y_i)$.

To formalize the correctness criteria of our simulation algorithm, we introduce \emph{legal configurations}, defined as those satisfying $(dfsstat, sim, smemupd)_i = (0, 0, 0)$ for all nodes $i \in V$, $sloc_v = 1$ for the current location $v$ of $a$, and $sloc_i = 0$ for all other nodes $i \neq v$.
Our algorithm satisfies the condition that the sequence of all legal configurations appearing during the execution of $A$ forms the simulated execution of $A^{\ast}$.
More precisely, let $E(G, v, A) = (C_t)_{t \geq 0}$ and $t_0, t_1, t_2, \dots$ be the maximal sequence of the round numbers such that $C_{t_i}$ is legal.
Then, our simulation algorithm is correct with respect to $T(i) = t_i$ and $\gamma=(\gamma_M,\gamma_A)$ such that $\gamma_M(M_i)=svars_i$ and $\gamma_A(M_i)=(smem_i, spin_i)$ hold.

Let $\mathcal{C}$ be a set of all valid legal configurations.
Our algorithm starts from any configuration $C\in \mathcal{C}$ and guarantees that the legal configuration $C'\in\mathcal{C}$ satisfies $\gamma(C)\xrightarrow{\phi^\ast,G}\gamma(C')$ without explicit termination.
Evidently, this algorithm provides simulated execution $E(G, v, A^{\ast})$:
Following the definition of the initial values specified in \cref{tb:variables}, the initial configuration $C_0=C_{t_0}$ of the corresponding execution $E(G, v, A)$ for any $v \in V$ is legal, and $\gamma(C_0)$ becomes the initial configuration of $E(G, v, A^\ast)$.
The algorithm then generates the configuration $C_{t_1}$ such that $\gamma(C_{t_0})\xrightarrow{\phi^\ast,G} \gamma(C_{t_1})$ holds.
The algorithm immediately starts the next round of simulations, which generates $C_{t_2}$ because $C_{t_1}$ is also legal.
$A$ simulates the execution of $E(G, v, A^{\ast})$ by repeating this process.

\subsection{Overviews of Our Algorithm}
\label{overview}
First, we explain the high-level structure of the proposed algorithm.
The simulator algorithm consists of five subtasks: local computation, DFS, clean-up, memory transfer, and move-and-reset phases.
We denote by $s$ the node where agent $a$ remains at the beginning of the one-round simulation.
Because the configuration at the beginning is legal, the simulated agent $a^{\ast}$ also stays at $s$ by definition.
We present an outline of each phase as follows:
\begin{description}
  \item[Local Computation Phase:]
  This phase was executed once at the beginning of the one-round simulation of $A^{\ast}$.
  The simulator agent locally simulates the local computation of $A^{\ast}$ at node $s$ (i.e., the location of $a^\ast$) and then updates the corresponding storage variables.
  After finishing the local computation phase, agent $a$ immediately proceeds to the DFS phase.
  \item[DFS Phase:] 
  Let $t$ be the node to which the simulated agent $a^{\ast}$ moves.
  Agent $a$ executes the DFS from $s$ with the first traversal edge $(s, t)$ until it revisits $s$ (i.e., the node such that the variable \textit{sloc} stores one).
  This determines a cycle containing edge $(s, t)$, which becomes the ITC.
  \item[Clean-up Phase:] 
  The agent repeats the same DFS traversal again from $s$ (with the first traversal edge $(s, t)$) to revisit $s$ to reset the storage values left in the DFS phase, except for the one necessary to recognize the ITC.
  \item[Memory Transfer Phase:]
  The agent circulates the ITC in the direction based on the memory value $smem_s$ of simulated agent $a^{\ast}$.
  \item[Move-and-Reset Phase:] 
  The agent moves simulated agent $a^{\ast}$ from $s$ to $t$ by setting $sloc_s=0$ and $sloc_t=1$.
  It then resets the storage of all nodes in the ITC to recover the legality of the configurations.
\end{description}
The agent determines the phase in which it is currently running based on the storage values $(dfsstat,sim,\allowbreak smemupd,sloc,par)$ of the current location and the current value of $p_{in}$.
The details of this are explained in the following section.

\subsection{Details of Our Algorithm}
\label{detail}
In this subsection, we present the details of each phase.
Pseudocodes are presented in \cref{alg:main,alg:dfs,alg:cleanup,alg:transmem,alg:movereset}.
\cref{alg:main} presents the main part, including the local computation phase.
We denote the procedures of the five phases above as \textit{LocalComp()}, \textit{DFS()}, \textit{CleanUp()}, \textit{TransMem()}, and \textit{MoveReset()}, respectively.
\begin{figure}[!t]
  \begin{algorithm}[H]
    \caption{Main Part and \textit{LocalComp()}}
    \label{alg:main}
    \begin{algorithmic}[1]
        \State $p_{in} \leftarrow$ (entry port number)
        \State $\Delta_i \leftarrow$ (the degree of the node $i$)
        \If{$smemupd_i=1$}
          \State \textit{MoveReset()}
        \ElsIf{$dfsstat_i>0$}
          \If{$sim_i=1\vee p_{in}=par_i\vee sloc_i=1$}
            \State \textit{CleanUp()}
          \Else
            \State \textit{DFS()}
          \EndIf
        \ElsIf{$sim_i=1$}
          \State \textit{TransMem()}
        \Else
          \If{$sloc_i=1$}
            \Comment{\textit{LocalComp()}}
            \State $(spout, svars, smem)_i \leftarrow \phi^{\ast}(\Delta_i, spin_i, svars_i, smem_i)$
          \EndIf
          \State \textit{DFS()}
        \EndIf
    \end{algorithmic}
  \end{algorithm}
\end{figure}
The procedure \textit{LocalComp()} is described in line 14 of \cref{alg:main} (see the comment in the pseudocode), and all other phases are separately described in \cref{alg:dfs,alg:cleanup,alg:transmem,alg:movereset}.
Hereinafter, we refer to a line $Y$ or lines $Y$-$Z$ of Algorithm $X$ as $(X,Y)$ or $(X,Y$-$Z)$.
In all the pseudocodes, the current location of agent $a$ is referred to as $i$.
As in \cref{overview}, we denote the initial location and destination of $a^{\ast}$ in the simulated round by $s$ and $t$, respectively.

\paragraph*{Main part and local computation phase}
The main part shown in \cref{alg:main} calls for an appropriate procedure according to the local storage value of node $i$.
The local computation phase is executed if $(dfsstat,sim,smemupd,\allowbreak sloc)_{i} = (0,0,0,1)$ holds.
As the initial configuration is legal, every node $i$ satisfies $(dfsstat,sim,\allowbreak smemupd)_i = (0,0,0)$.
We also have $sloc_s = 1$ and $sloc_i = 0$ for all $i \neq s$ based on the constraint of the variable \textit{sloc}.
According to the definition of \textit{sloc}, the current location of $a^{\ast}$ is also at $s$, and $smem_{s}$, $spin_{s}$, and $svars_{s}$ store the situation of agent $a^{\ast}$ in the simulated execution.
The simulator agent $a$ locally simulates the local computation of $A^{\ast}$ by $a^{\ast}$ at node $s$ and stores the computed results in $spout_{s}$, $svars_{s}$, and $smem_{s}$ (1,14).
The DFS phase was invoked immediately after the local computation phase (1,15).

\paragraph*{DFS phase}
We show the pseudocode for the DFS phase in \cref{alg:dfs}.
\begin{figure}[!t]
  \begin{algorithm}[H]
    \caption{\textit{DFS()}}
    \label{alg:dfs}
    \begin{algorithmic}[1]
        \State $lastin_i\leftarrow p_{in}$
        \If{$sloc_i=1$}
          \State $(dfsstat,cld,lastout)_i \leftarrow (1,spout_i,spout_i)$
        \ElsIf{$dfsstat_i=0$}
          \Comment{visiting by forward}
          \State $(dfsstat,par)_i \leftarrow (1,p_{in})$
          \If{$par_i = 0$}
            \Comment{skip the parent port}
            \State $(cld,lastout)_i \leftarrow (1,1)$
          \Else
            \State $(cld,lastout)_i \leftarrow (0,0)$
          \EndIf
        \ElsIf{$dfsstat_i=1\wedge p_{in}=cld_i$}
          \Comment{visiting by backtrack}
          \If{$par_i = p_{in}+1$}
            \Comment{skip the parent port}
            \State $cld_i \leftarrow p_{in}+2$
          \Else
            \State $cld_i \leftarrow p_{in}+1$
          \EndIf
          \If{$cld_i<\Delta_i$}
            \Comment{unprobed ports exist}
            \State $lastout_i\leftarrow cld_i$
          \Else
            \Comment{no unprobed port exists}
            \State $(lastout,dfsstat)_i\leftarrow (par_i,2)$
          \EndIf
        \Else
          \Comment{invoking backtrack}
          \State $lastout_i \leftarrow p_{in}$
        \EndIf
      \State Move to $\pi_{i}(lastout_i)$
    \end{algorithmic}
  \end{algorithm}
\end{figure}
This phase is executed if one of the following conditions is satisfied:
\begin{itemize}
  \item $(dfsstat,sim,smemupd,sloc)_i=(0,0,0,0)$ or immediately after the local computation phase (1,15),
  \item $dfsstat_i>0$, $(sim,smemupd,sloc)_i=(0,0,0)$, and $p_{in}\neq par_i$ (1,9).
\end{itemize}
Intuitively, the first and second conditions are applied to the cases in which the agent visits $i$ first and re-visits $i$.
Lines (2,2-3) are for the exceptional behavior of the first invocation of \textit{DFS()} immediately after the local computation phase.
Then, the choice of the probed port must correspond to edge $(s, t)$, that is, the port stored in $spout_i$.
Note that, when the agent visits node $i = s$ again, $dfsstat_i > 0$ and $sloc_i=1$ are satisfied, and (2,2-3) are not executed because the main routine invokes \textit{CleanUp()} in (1,7).
Because the simulation starts from a legal configuration, the variables $sim_i$, $smemupd_i$, and $sloc_i$ for any $i \neq s$ store 0 at the beginning of the DFS phase and are not modified in the subroutine \textit{DFS()}.
Consequently, the DFS phase continues if the agent visits nodes other than $s$.
Throughout the DFS phase, the agent visiting node $i$ stores the entry and outgoing port numbers in \textit{lastin} and \textit{lastout}, respectively.
Line (2,1) is for storing the entry port number.
To store the outgoing port number, the agent first writes the port number to which the agent will move into $lastout_i$ (lines 3, 7, 9, 16, 18, and 20) and finally moves to the port indicated by $lastout_i$ (2,21).
The information of \textit{lastin} and \textit{lastout} is used in procedure \textit{CleanUp()}.
During the DFS phase, the agent explored $G$ following the standard DFS.
Specifically, it probes the ports of the visited node individually in ascending order, which enables the agent to recognize the ports already probed by storing only one port number.
The latest probed port number was stored in $cld_i$.
When the agent returns to the current node $i$ through an edge with port number $cld_i$ by backtrack, $cld_i$ is incremented.
If $cld_i + 1$ is the port indicating the parent of the DFS tree, then it is skipped (2,6 and 2,11).
\footnote{Since the standard (centralized) DFS always moves to a neighboring unvisited node, agent-based algorithms cannot identify whether a neighbor is already visited.
Hence, the agent must check all the neighbors.
Then, the following case can occur: the agent exits from $i$ through port $p$, but the destination is already visited.
Thus, it returns immediately to $i$.
In our algorithm, this case is treated as a backtrack.}
If the value of $cld_i$ after the increment exceeds $\Delta_i - 1$, the agent has no neighbor of $i$ to be checked and thus performs the backtrack with setting $dfsstat_i = 2$ (2,18).
When the agent visits $i$ more than once, the variable $dfsstat_i$ stores a nonzero value at the second or later visit.
Then, the condition $(dfsstat, sim, smemupd)_i = (1, 0, 0)$ is satisfied, and the entry port number is necessarily different from $par_i$ because, in such a case, the agent traverses a non-DFS-tree edge or performs the backtrack.
Consequently, the agent can distinguish whether it visits $i$ by forward or non-forward movement under conditions $dfsstat_i=1$ and $p_{in} \neq par_i$ (2,10).

\paragraph*{Clean-up Phase}
We show the pseudocode of the clean-up phase in \cref{alg:cleanup}.
\begin{figure}[!t]
  \begin{algorithm}[H]
    \caption{\textit{CleanUp()}}
    \label{alg:cleanup}
    \begin{algorithmic}[1]
        \If{$sim_i=0$}
          \If{$(dfsstat,sloc)_i=(1,1)$}
            \State $(par,dfsstat,sim)_i \leftarrow (p_{in},0,1)$
            \State Move to $\pi_{i}(cld_i)$
          \Else
            \State $sim_i\leftarrow 1$
            \If{$par_i=0$}
              \Comment{skip the parent port}
              \State $cld_i\leftarrow 1$
            \Else
              \State $cld_i\leftarrow 0$
            \EndIf
            \If{$(dfsstat,lastin,lastout)_i=(1,p_{in},cld_i)$}
              \State $dfsstat_i\leftarrow 0$
              \Comment{initialization}
            \EndIf
            \State Move to $\pi_{i}(cld_i)$
          \EndIf
        \ElsIf{$p_{in}=cld_i$}
          \Comment{visiting by backtrack}
          \If{$par_i=p_{in}+1$}
            \Comment{skip the parent port}
            \State $cld_i\leftarrow cld_i +2$
          \Else
            \State $cld\leftarrow cld_i +1$
          \EndIf
          \If{$dfsstat_i=1$}
            \If{$(lastin,lastout)_i=(p_{in},cld_i)$}
              \State $dfsstat_i\leftarrow 0$
              \Comment{initialization}
            \EndIf
            \State Move to $\pi_{i}(cld_i)$
          \ElsIf{$cld_i<\Delta_i$}
            \Comment{unprobed ports exist}
            \State Move to $\pi_{i}(cld_i)$
          \Else
            \Comment{no unprobed port exists}
            \If{$(lastin,lastout)_i=(p_{in},par_i)$}
              \State $(dfsstat,sim)_i\leftarrow (0,0)$
              \Comment{initialization}
            \EndIf
            \State Move to $\pi_{i}(par_i)$
          \EndIf
        \Else
          \Comment{invoking backtrack}
          \If{$(lastin,lastout)_i=(p_{in},p_{in})$}
            \If{$dfsstat_i=1$}
              \State $dfsstat_i\leftarrow 0$
              \Comment{initialization}
            \Else
              \State $(dfsstat,sim)_i\leftarrow (0,0)$
              \Comment{initialization}
            \EndIf
          \EndIf
          \State Move to $\pi_{i}(p_{in})$
        \EndIf
    \end{algorithmic}
  \end{algorithm}
\end{figure}
This phase is executed if $dfsstat_i > 0$ holds, and either $sim_i =1$, $p_{in} = par_i$, or $sloc_i = 1$ is satisfied (1,5-6). 
This phase begins immediately after the agent returns to $s$ during the DFS phase.
At the beginning of this phase, the set of nodes that satisfy $dfsstat=1$ forms an ITC.
It first updates $par_i (=par_s)$ by $p_{in}$, resulting in consistent orientation of the ITC by the variable \textit{par}.
In addition, the variable \textit{cld} for all nodes in the ITC presents the inverse orientation.
The goal of this phase is to reset all the nodes $i$ with $dfsstat_i > 0$ by setting $dfsstat_i = 0$ and marking the nodes $i$ in the ITC with $sim_i = 1$.
The agent in the clean-up phase performs DFS up to the first revisit of $s$ similar to the DFS phase.
Every node $i$ visited in the clean-up phase must satisfy $dfsstat_i > 0$ (every node $i$ satisfying $dfsstat_i > 0$ is necessarily visited) because the DFS traversal in this phase is the same as that in the DFS phase.
The main technical challenge is to make the agent distinguish between the DFS  and clean-up phases, particularly at the first visit of $i$.
This issue was resolved as follows:
In the DFS phase, the agent never enters node $i$ with $dfsstat_i > 0$ (i.e., the second or later visit of $i$) through the edge with port number $par_i$.
Hence, if $p_{in}=par_i$ and $dfsstat_i>0$ holds in the clean-up phase, agent $a$ correctly recognizes that it enters $i$ by the forward movement of the DFS in the clean-up phase.
The value update of $sim_i$ from zero to one occurs when the agent visits $i$ first in the clean-up phase (3,6), which enables the agent to distinguish between the clean-up and DFS phases when revisiting $i$.
During the execution, the agent must initialize the information of the storages written in the DFS phase, except for the information for the ITC (i.e., the information of $sim_i = 1$ for $i$ in the ITC).
However, resetting $dfsstat_i$ can result in the loss of recognition in the current phase when the agent visits $i$ again in the subsequent execution of the clean-up.
To avoid this, the agent checks whether the current visit of $i$ is the final one by comparing the number of entry and outgoing ports to the values of $(lastin, lastout)_i$ left in the DFS phase.
One can show that the agent never visits $i$ again in the subsequent execution of the clean-up phase if they are the same (formally proved as \cref{clm:lastinout}).
Finally, when the agent returns to $s$, all unnecessary information is reset, and the agent proceeds to the memory transfer phase under the condition $(dfsstat,sim,smemupd)_s = (0,1,0)$ (1,10).
It is noteworthy that $sim_j=1$ holds if and only if $j$ is contained in the constructed ITC.

\paragraph*{Memory Transfer Phase}
We show the pseudocode of the memory transfer phase in \cref{alg:transmem}.
\begin{figure}[!t]
  \begin{algorithm}[H]
    \caption{\textit{TransMem()}}
    \label{alg:transmem}
    \begin{algorithmic}[1]
        \If{$sloc_i=1$}
          \State $smemupd_i\leftarrow 1$
          \If{$smem_i=0$}
            \Comment{transfer 0 ($p_{in}=par_i$)}
            \State Move to $\pi_{i}(cld_i)$
          \Else
            \Comment{transfer 1 ($p_{in}=cld_i$)}
            \State Move to $\pi_{i}(par_i)$
          \EndIf
        \Else
          \Comment{receive the transferred value}
          \State $smemupd_i\leftarrow 1$
          \If{$p_{in}=par_i$}
            \Comment{transfer 0}
            \State $smem_i\leftarrow 0$
            \State Move to $\pi_{i}(cld_i)$
          \Else
            \Comment{transfer 1}
            \State $smem_i\leftarrow 1$
            \State Move to $\pi_{i}(par_i)$
          \EndIf
        \EndIf
    \end{algorithmic}
  \end{algorithm}
\end{figure}
In this phase, the agent circulates ITC.
Recall that node $i$ in the ITC is identified by condition $sim_i = 1$.
The agent recognizes that the current round is in the memory transfer phase if $(sim, dfsstat)_i = (1, 0)$ holds (1,10-11).
It sets $smemupd_i=1$ and moves to one of $i$'s neighbors in the ITC, specified by the value of $smem_i$.
As stated above, the variables \textit{par} and \textit{cld} at the nodes in the ITC present two opposite orientations of the ITC; thus, we observe that agent $a$ transfers 0 if $p_{in}=par_i$ and 1 otherwise.
The transferred value is written for  all nodes in the ITC (4,10 and 4,13).
If $smemupd_i = 1$ at the visited node, the agent detects the termination of circulating the ITC, which triggers the following move-and-reset phase (1,3-4).

\paragraph*{Move-and-Reset Phase}
We show the pseudocode for the move-and-reset phases in \cref{alg:movereset}.
\begin{figure}[!t]
  \begin{algorithm}[H]
    \caption{\textit{MoveReset()}}
    \label{alg:movereset}
    \begin{algorithmic}[1]
        \If{$sloc_i=1$}
        \Comment{leave from $s$}
          \State $sloc_i\leftarrow 0$
          \State Move to $\pi_{i}(spout_i)$
        \ElsIf{$p_{in}=par_i$}
          \Comment{reach to $t$}
          \State $(sloc,spin,sim,smemupd)_i\leftarrow (1,p_{in},0,0)$
          \State Move to $\pi_{i}(par_i)$
        \Else
          \Comment{initialization}
          \State $(sim,smemupd)_i\leftarrow (0,0)$
          \State Move to $\pi_{i}(par_i)$
        \EndIf
    \end{algorithmic}
  \end{algorithm}
\end{figure}
This phase is executed when $smemupd_i = 1$ (1,3). 
First, agent $a$ sets $sloc_s=0$ at node $s$, which implies that $a^\ast$ leaves $s$ and moves to $\pi_s(spout_s)=t$.
At node $t$, the agent updates $sloc_t$ and $spin_t$ by 1 and $p_{in}$, respectively (5,5).
The one-round simulation of $a^{\ast}$ finishes because $smem_t$ is already updated by the memory transfer phase.
The remaining task is to reset the expired information of variables \textit{sim} and \textit{smemupd} left in the ITC.
To reset this, the agent circulates the ITC following the orientation by \textit{cld}.
The requirement for the orientation is to make the agent distinguish the first movement (for updating \textit{sloc} and \textit{spin}) from the following reset movement.
Since the port number of the edge $(s, t)$ at $t$ is $par_t$, if the agent visits $t$ with $p_{in} = par_t$, it can recognize that the current round is for updating \textit{sloc} and \textit{spin} (5,4), in contrast with the fact that $p_{in} = cld_t$ always holds in the reset circulation.
If $sloc_i = 1$ holds in the reset circulation, this implies that $i = t$, and $(dfsstat,sim,smemupd)_i=(0,0,0)$ also holds for all $i$.
The configuration becomes legal since all nodes not in the ITC have been reset correctly in the clean-up phase.

\subsection{Correctness}
In this section, we prove that each procedure correctly executes the corresponding phase and bound their running times.
To prove this, we first introduce the standard (centralized) DFS process under the following conditions:
\begin{enumerate}
  \item Initially, the search head points at $s$ and moves to $t$ at the first neighborhood search.
  \item When visiting $i \neq s$, the search head sequentially probes its neighbors along the order of the corresponding port numberings.
  \item If the search head moves to a node $i$ already searched, it immediately goes back to the previous node.
  \item The search process terminates when the search head reaches $s$ again.  
\end{enumerate}
where $s$ and $t$ denote the nodes defined in the previous section.
We define $V_{dfs}$ as the set of all nodes reached by the search head in this process and $S=i_0, i_1, \dots i_M$ as the trajectory of the search head.
Note that, in this sequence, one node can appear more than once.
In addition, for node $s'$ visited immediately before re-vising $s$, we define $V_{ITC}$ as the set of nodes on the path from $s$ to $s'$ in the DFS tree.
All $V_{dfs}$, $S$, and $V_{ITC}$ are uniquely determined only by nodes $s$ and $t$, graph $G$, and its port numbering functions $\Pi$.

For any phase $X$, if the agent invokes the procedure corresponding to $X$ in round $r$, we say that round $r$ is of phase $X$.
To make it well-defined, we slightly modified our algorithm by separating the round executing \textit{LocalComp()} and executing \textit{DFS()}.
That is, after the local computation phase, the agent finishes the current round with no movement, and in the next round, it invokes \textit{DFS()}.
Trivially, the correctness of the modified algorithm was lower than that of the original algorithm.
Let $r_X$ be the first round of phase $X$ during the execution.
The \emph{sub-execution} of phase $X$ is defined as the maximal period from $r_X$ in which every round is of phase $X$.
For the final round $r'$ of the sub-execution of $X$, the configuration at the beginning of round $r' + 1$ is called the \emph{resultant configuration} of $X$.  

A node $i$ is called \emph{legal} if it satisfies $(dfsstat, sim, smemupd)_i = (0, 0, 0)$.
First, we present the correctness criteria for each phase.
\begin{definition}[correctness of phases]
  \label{def:correctness}
  Each phase is called \emph{correct} if the sub-execution of the phase satisfies the following condition:
  \begin{description}
    \item[Local Computation Phase]
    Starting from any valid legal configuration $C\in \mathcal{C}$, the resultant configuration $C'$ satisfies
    $(spout,svars,smem)_{s} = \phi^\ast(\Delta_{l^\ast},p_{in}^\ast,M^\ast_{l^\ast},m^\ast)$ and $C' \in\mathcal{C}$.
    \item[DFS Phase]
      Starting from the resultant configuration of any correct local computation phase, the sub-execution reaches the resultant configuration satisfying as follows:
      \begin{enumerate}
        \item For any $u \in V$, $(sim, smemupd)_u = (0, 0)$ holds.
        \item For any $u \in V$, $dfsstat_u = 1$ holds if $u \in V_{ITC}$, $dfsstat_u = 2$ holds if $u \in V_{dfs} \setminus V_{ITC}$, or $dfsstat_u = 0$ holds otherwise.
        \item The variables $(spout, svars, smem)_s$ store the same values as the starting configuration.
      \end{enumerate}
    \item[Clean-up Phase]
      Starting from the resultant configuration of any correct DFS phase, the sub-execution reaches the resultant configuration satisfying as follows:
      \begin{enumerate}
        \item For each $u \in V_{ITC}$, $(dfsstat, sim, smemupd)_u = (0, 1, 0)$ holds.
        \item For each $u \in V_{ITC}$, $\pi_u(par_u)$ and $\pi_u(cld_u)$ are also contained in $V_{ITC}$, and letting $\pi_u(par_u) = x$, $\pi_x(cld_x) = u$ is satisfied.
        \item Any node $u \notin V_{ITC}$ is legal, and $(spout, svars, smem)_s$ store the same values as the starting configuration.
      \end{enumerate}
    \item[Memory Transfer Phase]
      Starting from the resultant configuration of any correct clean-up phase, the sub-execution reaches the resultant configuration satisfying as follows:
      \begin{enumerate}
        \item For any $u \in V_{ITC}$, $(dfsstat,sim,smemupd)_u=(0,1,1)$ and $smem_u=smem_s$ holds.
        \item Any node $w \notin V_{ITC}$ is legal, and $(spout, svars, smem)_s$ store the same values as the starting configuration.
      \end{enumerate}
    \item[Move-and-Reset Phase]
      Starting from the resultant configuration of any correct memory transfer phase, the sub-execution reaches the resultant configuration satisfying as follows: 
      \begin{enumerate}
        \item Any node $i\in V$ is legal, and $sloc_t=1$ holds only for $t$.
        \item The variable $spin_t$ stores $\pi^{-1}_t(s)$ (i.e., the entry port number of $a^\ast$). 
        \item $smem_t=smem_s$.
      \end{enumerate}
  \end{description}
\end{definition}
It is easy to verify that the resultant configuration of the phase triggers the subsequent phase.
Hence, one can trivially deduce the correctness of our algorithm based on the correctness of all phases.
We focus on the correctness of the remaining four phases because procedure \textit{LocalComp()} is trivially correct.
\begin{lemma}
  \label{lem:dfs}
  The DFS phase is correct.
\end{lemma}
\begin{proof}
  Procedure \textit{DFS()} does not modify \textit{sim} or \textit{smemupd}.
  Because we assume that the initial configuration is legal, this proof does not consider the conditions of \textit{sim} and \textit{smemupd}, but focuses on the conditions of variable \textit{dfsstat}.
  We first show that the agent moves along the trajectory $i_0, i_1, \dots, i_M$ during the sub-execution of the DFS phase.
  If the procedure \textit{DFS()} is invoked in some rounds, it chooses the next visited node in the same way as the centralized DFS explained above.
  Hence, it suffices to show that the DFS phase continues until the round when the agent visits $s$.
  If $dfsstat_i=0$ holds, procedure \textit{DFS()} is necessarily executed (1,12 and 1,15).
  In addition, during the DFS phase, $sim_i=0$ and $p_{in}\neq par_i$ hold for any $i\in V$ such that $dfsstat_i > 0$ because $dfsstat_i > 0$ implies that $i$ is already visited in the DFS phase.
  Thus, $par_i$ appropriately points to the parent of $i$ in the DFS tree (evidently, the DFS search never visits node $i$ twice through the edge from the parent).
  Thus, the sub-execution of the DFS phase terminates if and only if $sloc_i=1$ and $dfsstat_i>0$ hold (1,5-7).
  When agent $a$ revisits $s$, $dfsstat_s=1$ and $sloc_s=1$ hold.
  Hence, the DFS phase continues until node $s$ has been revisited.
  We show that the resultant configuration satisfies the condition of \cref{def:correctness}.
  Let $s'$ be the last newly visited node in the DFS phase (i.e., the node immediately before the revisit of $s$). 
  Variable $dfsstat_i$ is set to 2 (and the agent backtracks) if and only if all ports except that $par_i$ are probed.
  According to the definition of $V_{ITC}$, the agent does not backtrack at any node in $V_{ITC}$.
  That is, at any node $u \in V_{ITC}$, the agent does not return to $u$ through the edge indicated by $cld_u$.
  This implies that $dfsstat_u = 1$ holds and the condition of \cref{def:correctness} is satisfied at any $u \in V_{ITC}$.  
\end{proof}

The proof of the clean-up phase contains a slight technical non-triviality in the distinction of the DFS phase.
\begin{lemma}
  The clean-up phase is correct.
\end{lemma}
\begin{proof}
  As a proof of \cref{lem:dfs}, we do not consider the condition of the variable \textit{smemupd}.
  In this phase, the agent traces the same route as in the DFS phase if resetting \textit{dfsstat} and \textit{sim} at node $i$ occurs correctly at the last visit of $i$. 
  We prove this by inducing  the indices of trajectory $i_0,i_1,\dots i_M$.
  Let $c_0,c_1,\dots c_{M'}$ be the trajectory of the agent in the clean-up phase.
  We show that $c_k=i_k$ holds and $c_k$ satisfies the conditions for executing the clean-up phase for any $k$.
  Because this phase starts at node $s$, $c_0=i_0$ and the condition for executing the clean-up phase is satisfied.
  Then, as the induction hypothesis, we assume that $c_k=i_k$ is satisfied for each $k\in [0,j]$ and the condition for executing the clean-up phase is satisfied at $c_j$.
  Because the method of deciding the next probed neighbor is the same as the DFS phase, $c_{j+1}=i_{j+1}$ is satisfied.
  Now, we have to show that the condition for executing the clean-up phase is satisfied at $c_{j+1}$.
  If $c_{j+1}$ is visited for the first time in this phase, $p_{in}=par_{c_{j+1}}$ holds, and the agent sets $sim_{c_{j+1}}=1$.
  Otherwise, the proof is performed if $sim_{c_{j+1}}=1$ holds; that is, $c_{j+1}$ is not reset up to round $j$.
  As we have already shown in \cref{lem:dfs} that $i_0, i_1, \dots, i_M$ is also the trajectory of the agent in the DFS phase, the agent stores the entry/outgoing port numbers in the variables \textit{lastin} and \textit{lastout} of node $i_k$ in the $k$-th round of the DFS phase.
  Let $L_k$ be a pair of values stored in $(lastin, lastout)_{i_k}$ when visiting $i_k$.
  Then, we have the following claim.
  \begin{claim}
    \label{clm:lastinout}
    For any $k',k \in [0, M]$, $L_{k'} \neq L_{k}$ holds, if $i_k=i_{k'}$ holds.
  \end{claim}
  \begin{proof}
    Assume $k' < k$ without loss of generality, and let $u = i_{k'}$ for short.
    We denote the values of \textit{lastin} and \textit{lastout} in $L_x$ as $lastin(x)$ and $lastout(x)$ respectively.
    First, consider the case in which round $k'$ is the first visit of $u$ in the DFS process.
    Subsequently, $lastin(k') = par_u$ holds.
    Because the agent visits $u$ from its parent only at the first visit, we have $lastin(k) \neq par_u$, and thus the claim holds.
    Next, consider the case in which round $k'$ is not the first visit of $u$.
    If $lastin(k') = lastin(k)$ holds, then the agent visits $u$ in rounds $k'$ and $k$ through the same edge $(w, u)$.
    There exist two scenarios in which the agent traverses edge $(w, u)$ in the second or later visit of $u$:
    In the first scenario, the agent moves forward from $u$ to $w$, but $w$ has already been visited and thus returns to $u$ by backtracking.
    Second, the agent moves forward from $w$ to $u$.
    These two are all possible scenarios, and thus they occur at rounds $k' - 1$ and $k - 1$, respectively.
    While the agent leaves from $u$ with the port necessarily different from $\pi^{-1}_u(w)$ in the first scenario, it leaves with the edge to $w$ in the second scenario (because the next movement is backtracked to $w$).
    That is, the two scenarios choose different numbers of outgoing ports to leave $u$.
    Thus, $lastout(k') \neq lastout(k)$ holds.
    The claim holds.
  \end{proof}
  Suppose that, for contradiction, $c_{j+1}$ is reset in round $k \leq j$.
  Because $c_{j + 1} = i_{j+1}$ holds, round $k$ is not the last visit of $c_k = c_{j+1}$.
  Let $c_{k'}$ be the last visit of $c_{j+1} (k' > k)$, and \cref{clm:lastinout} implies that $L_{k'} \neq L_{k}$.
  This contradicts the assumption that $c_{j+1}$ is reset to the round $k$.
  Consequently, $c_{j+1}$ is not reset to round $j$ and it can be concluded that the trace of the agent is the same as that of  the DFS phase.
  In addition, it has been proven that reset at node $i$ occurs at the last visit of $i$.
  From these two facts, it is easy to verify that the resultant configuration satisfies the condition of \cref{def:correctness}.
  The lemma is proved.
\end{proof}

The correctness of the memory transfer and the move-and-reset phases is relatively straightforward.
\begin{lemma}
  The memory transfer phase is correct.
\end{lemma}
\begin{proof}
  Under the condition of the initial configuration, the nodes $i$ with $sim_i = 1$ induce ITC, and variables \textit{par} and \textit{cld} provide two opposite orientations.
  Hence, evidently, the agent correctly circulates the ITC in the memory transfer phase and terminates at the initial position $s$ (note that, at the beginning of the phase, the agent sets $smemupd_s = 1$; thus, when it returns to $s$, the memory transfer phase necessarily terminates).
  Assuming that variables \textit{par} and \textit{cld} provide two opposite orientations if the agent leaves $s$ through port $par_s$, it enters the next node $x$ through port $cld_x$ and vice versa.
  This observation implies that the simulated memory value $smem_s$ is correctly transferred to all nodes in the ITC, and the condition of the resultant configuration in \cref{def:correctness} is satisfied.
\end{proof}

\begin{lemma}
  The move-and-reset phase is correct.
\end{lemma}
\begin{proof}
  Under the condition of the initial configuration, all nodes not in $V_{ITC}$ have already been legal (i.e., $(dfsstat, sim, smemupd) = (0, 0, 0)$). 
  As explained in \cref{detail}, the agent correctly circulates the ITC along the orientation using \textit{cld} and resets \textit{sim} and \textit{smemupd}.
  Because \textit{dfsstat} is already zero, we can conclude that the resulting configuration is legal.
  The remaining issue is to show that it is the configuration interpreted as that for $A^{\ast}$ after a one-round simulation.
  It has already been shown that the memory value is transferred correctly and $svars_s$ is updated correctly in the local computation phase.
  In addition, the agent updates $spin_t$ by $p_{in}$ at the first movement from $s$ to $t$ in the move-and-reset phase; that is, $spin_t$ stores the port number of edge $(s, t)$.
  Because $a^\ast$ moves from $s$ to $t$ in this round, $spin_t$ is equal to the entry port number of $a^\ast$ in the next simulated round.
\end{proof}

By combining all of the lemmas above, we obtain the main theorem. 
\begin{theorem}
  Algorithm $A$ shown in \cref{alg:main} correctly simulates $A^{\ast}$ in any 2-edge-connected graph $G$.
  The round complexity of the algorithm 
  to simulate one round of $A^{\ast}$ is $O(n^2)$.
  The required storage size is $O(\lambda^{\ast} + \log \Delta)$ bits per node.
\end{theorem}
\begin{proof}
  The correctness clearly follows \cref{def:correctness}, and all phases are correct.
  To limit time complexity, we bound the running time of each phase.
  The first two phases (i.e., DFS and clean-up) depend on the running time of DFS.
  To execute DFS for the graph such that the number of nodes and edges are $\left\lvert V \right\rvert$ and $\left\lvert E \right\rvert$, respectively, it needs $O(\left\lvert E\right\rvert)=O(n^2)$ rounds.
  The running time of the remaining two phases was bounded by $O\left(\left\lvert V_{ITC} \right\rvert\right)=O(n)$ rounds.
  Hence, the total running time (per round of simulation) is $O(n^2)$. 

  The storage size consumes $\lambda^\ast$ bits for \textit{svars}, which corresponds to the storage of $A^\ast$.
  The sizes of the other variables are bounded by $O(\log \Delta)$.
  The storage size was $O(\lambda^{\ast} + \log \Delta)$ bits per node.
\end{proof}

\section{Conclusion}
In this study, we demonstrated the equivalence of the computational power of a single oblivious agent and a single one-bit memory agent in a 2-edge-connected graph with $O(\log \Delta)$-bit storage by presenting the algorithm $A\in \mathcal{A}(0,\lambda)$ that simulates any $A^\ast\in \mathcal{A}(1,\lambda^\ast)$ on 2-edge-connected graphs with $\lambda=\lambda^\ast+O(\log\Delta)$.
The time overhead of our simulation algorithm was $O(n^2)$ per round.
That is, if the original algorithm $A^{\ast}$ runs in polynomial time, then our simulator also works in polynomial time. 
Finally, we conclude this study by explaining the open problem.
The primary open problem deduced from our results is whether an oblivious agent can simulate any one-bit agent on 2-edge-connected graphs under the assumption that only $O(1)$-bit storage per node is available.
As a weaker form of this problem, 2-edge-connected graphs with $O(1)$-bit storage per node can be also explored.

\section*{Acknowledgement}
This study was initiated by the discussion with Prof. Shantanu Das when he visited the third author (Taisuke Izumi) in 2018.
We greately appreciate his valuable comments at the discussion.

\bibliography{SimulationOn2-ConnectedGraph}
\bibliographystyle{plain}

\end{document}